\documentclass[11pt]{amsart}
\usepackage{graphicx, amssymb}
\usepackage{amsmath}

\newtheorem{thm}{Theorem}[section]

\newtheorem{lem}[thm]{Lemma}

\theoremstyle{definition}

\theoremstyle{remark}
\newtheorem{rem}[thm]{Remark}

\numberwithin{equation}{section}
\newcommand{\abs}[1]{\left\vert#1\right\vert}
\newcommand{\set}[1]{\left\{#1\right\}}
\newcommand{\Real}{\mathbb R}
\newcommand{\Natural}{\mathbb N}

\newcommand{\B}{\mathcal{B}}
\newcommand{\such}{\, | \,}
\newcommand{\limn}{\lim_{n \to \infty}}
\newcommand{\limk}{\lim_{k \to \infty}}
\newcommand{\dfn}{\, := \,}

\newcommand{\prob}{\mathbb{P}}
\newcommand{\proby}{\prob^Y}
\newcommand{\tprob}{\widetilde{\mathbb{P}}}

\newcommand{\Exp}{\mathcal E}
\newcommand{\qprob}{\mathbb{Q}}
\newcommand{\qprobb}{\mathsf{Q}}
\newcommand{\expec}{\mathbb{E}}

\newcommand{\F}{\mathcal{F}}

\newcommand{\ud}{\mathrm d}
\newcommand{\udw}{\mathrm d}

\newcommand{\X}{\mathcal{X}}

\newcommand{\rel}{\mathsf{rel}}
\newcommand{\rely}{\rel^Y}

\newcommand{\normTV}[1]{\big| #1 \big|_{\mathsf{TV}}}

\newcommand{\p}{\mathrm{p}}

\newcommand{\tX}{\widetilde{X}}

\newcommand{\hX}{\widehat{X}}

\newcommand{\tp}{\widetilde{p}}
\newcommand{\tvartheta}{\widetilde{\vartheta}}

\newcommand{\co}{\mathsf{c}}

\newcommand{\g}{\mathsf{g}}
\newcommand{\supp}{\mathsf{supp}}
\newcommand{\cosupp}{\mathsf{conv.supp}}
\newcommand{\gy}{\g^Y}
\newcommand{\ngo}{\nabla \g}
\newcommand{\ngoy}{\ngo^Y}

\newcommand{\pare}[1]{\left(#1\right)}
\newcommand{\bra}[1]{\left[#1\right]}
\newcommand{\dbra}[1]{[\kern-0.15em[ #1 ]\kern-0.15em]}
\newcommand{\dbraco}[1]{[\kern-0.15em[ #1 [\kern-0.15em[}

\newcommand{\indic}{\mathbb{I}}

\newcommand{\ELMD}{\emph{ELMD}}

\newcommand{\absco}{{<\kern-0.53em<}}

\newcommand{\NAone}{\emph{NA}$_1$}

\begin{document}

\title[Market viability via absence of arbitrage of the first kind]{Market viability via absence of arbitrage of the first kind}%
\author{Constantinos Kardaras}%
\address{Constantinos Kardaras, Mathematics and Statistics Department, Boston University, 111 Cummington Street, Boston, MA 02215, USA.}%
\email{kardaras@bu.edu}%

\thanks{The author would like to thank Yuri Kabanov for fruitful conversations that significantly helped in formulating and proving the results of this paper. Two anonymous referees provided invaluable help in the presentation of the paper. Partial support by the National Science Foundation, grant number DMS-0908461, is gratefully acknowledged.}%
%\subjclass[2000]{60H99, 60G44, 91B28, 91B70}
\keywords{Arbitrage of the first kind, cheap thrills, fundamental theorem of asset pricing, equivalent local martingale deflators, semimartingales, predictable characteristics.}%
\subjclass[2000]{60G44, 60H99, 91B28, 91B70}
\date{\today}%
%\dedicatory{}%
%\commby{}%
% ----------------------------------------------------------------
\begin{abstract}
In a semimartingale financial market model, it is shown that there is equivalence between absence of \textsl{arbitrage of the first kind} (a weak viability condition) and the existence of a strictly positive process that acts as a local martingale deflator on nonnegative wealth processes.
\end{abstract}

\maketitle

% ----------------------------------------------------------------
\setcounter{section}{-1}

\section{Introduction}

A ubiquitous market assumption in the literature of Stochastic Finance theory is postulating the existence of Equivalent Local Martingale Measure (ELMM). The latter refers to a probability measure $\qprob$, equivalent to the ``real-world'' probability $\prob$, with the property that all discounted nonnegative wealth processes are local $\qprob$-martingales. In view of the Fundamental Theorem of Asset Pricing (FTAP), it is quite clear why such assumption is made from the outset: existence of an ELMM is intimately connected to market viability; in fact, it is equivalent to the economically-sound ``No Free Lunch with Vanishing Risk'' (NFLVR) condition --- see for example \cite{MR1304434} and \cite{MR1671792} for a complete treatment on the topic.

\smallskip

Stipulating the existence of an ELMM seems unavoidable in order to maintain market viability. However, in recent publications there has been considerable interest in models where an ELMM might fail to exist. These have appeared, for instance:
\begin{itemize}
  \item in the context of stochastic portfolio theory, for which the survey \cite{FerKar07} is a good introduction;
  \item from the financial modeling perspective, an example of which is the \textsl{benchmark approach} of \cite{MR2267213};
  \item in a financial equilibrium setting, both for infinite-time horizon settings (see \cite{RePEc:ier:iecrev:v:33:y:1992:i:2:p:323-39}),
as well as finite-time horizon models with credit constraints on economic agents (see \cite{MR1774056} and \cite{MR1748373}).
\end{itemize}
The common assumption that the previous approaches share is postulating the existence of an Equivalent Local Martingale Deflator (ELMD), that is, a strictly positive process that makes all discounted nonnegative wealth processes, when multiplied by it, local martingales. (An ELMD was called a \textsl{strict martingale density} in \cite{MR1353193}; we opt here to call it ELMD as it immediately connects with the notion of an ELMM.) An ELMD is a strictly positive local martingale, but not necessarily a martingale; therefore, it cannot always be used as a density processes to produce an ELMM. 

\smallskip

While models where an ELMM might fail to exist are now being extensively studied, a result that would justify their applicability along the lines the FTAP has not yet appeared in the literature. %Such absence of a theoretical foundation could be partially responsible for the lack of enthusiasm in accepting these models.
In this work, the aforementioned issue is tackled. A precise economical condition of market viability is given using the concept \textsl{arbitrage of the first kind}, which has first appeared under this appellation in \cite{Ing87}; see also \cite{MR1348197} in the context of large financial markets, as well as \cite{MR1774056}, where arbitrage of the first kind is called a \textsl{cheap thrill}. Absence of arbitrage of the first kind in the market, which we shall abbreviate as condition NA$_1$, is close in spirit, but strictly weaker, than condition NFLVR; in fact, it is exactly equivalent to condition ``No Unbounded Profit with Bounded Risk'' (NUPBR) that appeared in \cite{MR2335830}. The main result of the present paper precisely states that in a semimartingale market model there is  equivalence between condition NA$_1$ and the existence of an ELMD.

%\smallskip
%
%The starting point of assuming condition NA$_1$ allows for more financial modeling freedom, as it expands the ``allowable'' class of models under the classical theory. Not requiring the existence of an ELMM has a number of advantages, both from practical and theoretical considerations. There has already been quite some discussion on this topic already. In order not to be repetitive, no attempt shall be made here to justify the usefulness of such approach; we rather refer the interested reader to \cite{FerKar07}, \cite{MR2335830}, \cite{KarPla07} and the references therein. The author's hope is that the present work will go one step further in popularizing models where ``free snacks'' in the terminology of \cite{MR1774056} might exist, by providing a justification that parallels the FTAP. Needless to say, the appropriateness of choosing this perspective as an alternative to the classical modeling assumption of existence of an ELMM depends on the problem-in-hand.
%
%\cite{MR1660801}

\smallskip

In the literature concerning discrete-time models, there have appeared two ways of providing a proof of the FTAP. The first one is the approach of \cite{MR1041035} (initiated in \cite{MR540823}), which utilizes convex separation functional-analytic arguments. The alternative, presented in \cite{MR1380761}, uses the economic idea that the marginal utility evaluated at the optimal terminal wealth of an economic agent, when properly scaled, defines the density of an equivalent martingale measure. The former approach has been adapted with extreme success to continuous time models in \cite{MR1304434} and \cite{MR1671792}. The present work can be seen as a counterpart of the latter approach in continuous-time markets --- here, the utility involved is logarithmic (under a suitable change of probability), and makes the reciprocal of the log-optimizer an ELMD. Interestingly enough, in continuous-time models the two approaches do not give rise to the same result; the present approach weakens the equivalent conditions of the  classical FTAP in \cite{MR1304434}, both from the mathematical and the economic side. Note also that the main result of this paper can also be seen as an intermediate step in proving the general version of the FTAP as is presented in \cite{MR1304434}. In fact, this task is taken up in \cite{Kar_09_fin_add_ftap}.

\medskip

The structure of the paper is simple. In Section \ref{sec: weak version of FTAP}, the market is introduced, arbitrage of the first kind is defined and the main result is stated, whose somewhat lengthy and technical proof is deferred for Section \ref{sec: proof}.

\section{Absence of Arbitrage of the First Kind and Equivalent Local Martingale Deflators} \label{sec: weak version of FTAP}

\subsection{Probabilistic remarks}
All stochastic processes in the sequel are defined on a \textsl{filtered probability space} $\left(\Omega, \, \F, \, (\F(t))_{t \in \Real_+}, \, \prob\right)$. Here, $\prob$ is a probability on $(\Omega, \F)$, $\F$ being a sigma-algebra that will make all random elements measurable. All relationships between random variables are understood in the $\prob$-a.s. sense. The filtration $(\F(t))_{t \in \Real_+}$ is right-continuous. We assume the existence of a finite financial planning horizon $T$, where $T$ is a \emph{finite} stopping time. All processes will be assumed to be constant, and equal to their value they have at $T$, after time $T$. Without affecting the generality of the discussion, it will be assumed throughout that $\F(0)$ is trivial modulo $\prob$ and that $\F(T) = \F$.

\subsection{Investment}

Let $S$ be a \emph{semimartingale}, denoting the \emph{discounted}, with respect to some baseline security, price process of a financial security. Starting with capital $x \in \Real$, and investing according to some predictable and $S$-integrable strategy $\vartheta$, an economic agent's discounted wealth process is
\begin{equation} \label{eq: wealth process, all}
X^{x, \vartheta} \dfn x + \int_0^\cdot \vartheta(t) \ud S(t).
\end{equation}
When modeling frictionless trading, credit constraints have to be imposed on investment in order to avoid \emph{doubling strategies}. Define then $\X$ to be the set of all nonnegative wealth processes, i.e., all $X^{x, \vartheta}$ in the notation of \eqref{eq: wealth process, all} such that $X^{x, \vartheta} \geq 0$.

\subsection{Equivalent local martingale deflators}

An \textsl{equivalent local martingale deflator} (ELMD) is a nonnegative process $Z$ with $Z(0) = 1$ and $Z(T) > 0$, such that $Z X$ is a local martingale for all $X \in \X$. Since $1 \equiv X^{1, 0} \in \X$, an ELMD is in particular a strictly positive local martingale.

\subsection{Arbitrage of the first kind}

An $\F(T)$-measurable random variable $\xi$ will be called an \textsl{arbitrage of the first kind} if $\prob[\xi \geq 0] = 1$, $\prob[\xi > 0] > 0$, and \emph{for all $x > 0$ there exists $X^{x, \vartheta} \in \X$ (for some $\vartheta$ which may depend on $x$), such that $X^{x, \vartheta} (T) \geq \xi$}. If there exists no arbitrage of the first kind in the market, we shall say that condition NA$_1$  holds.

It is straightforward to see that condition NA$_1$ is weaker than condition NFLVR of \cite{MR1304434}. Actually, using a combination of Lemma A.1 in \cite{MR1304434} and Lemma 2.3 in \cite{MR1768009}, it is shown in \cite[Proposition 1.2]{Kar_09_fin_add_ftap} that condition NA$_1$ is equivalent to the requirement that the set $\{X (T) \such  X \in \X \text{ with } X_0 = 1\}$ is bounded in probability. The latter condition has been coined BK in \cite{MR1647282} and NUPBR in \cite{MR2335830}.

\subsection{The main result} The next result can be seen as a weak version of the FTAP. Though simple to state, its proof is quite technical and is given in Section \ref{sec: proof}.

\begin{thm} \label{thm: main}
Condition \NAone \ is equivalent to the existence of at least one \ELMD.
\end{thm}

\begin{rem}
In \cite{Kar_09_fin_add_ftap}, which is in a certain sense a sequel to this paper, it is argued that although an ELMD does not generate a probability measure, its local martingale structure allows one to define a \emph{finitely additive} probability that is \emph{locally countably additive} and \emph{weakly equivalent} to $\prob$, and further makes discounted asset-price processes behave like ``local martingales''. More precisely, Theorem \ref{thm: main} can be reformulated to state that condition NA$_1$ is valid if and only if there exists $\qprobb : \F \mapsto [0,1]$ and a a sequence $(\tau_n)_{n \in \Natural}$ of stopping times with $\limn \prob \bra{\tau_n = T} = 1$ such that:
\begin{itemize}
  \item $\qprobb[\emptyset] = 0$, $\qprobb[\Omega] = 1$, and $\qprobb$ is (finitely) additive: $\qprobb[A \cup B] = \qprobb[A] + \qprobb[B]$ whenever $A \in \F$ and $B \in \F$ satisfy $A \cap B = \emptyset$;
  \item for $A \in \F$, $\prob[A] = 0$ implies $\qprobb[A] = 0$;
  \item when restricted on $\F_{\tau_n}$, $\qprobb$ is countably additive and equivalent to $\prob$, for all $n \in \Natural$.
  \item $\int_{\Omega} X_{\tau^n \wedge \tau} \ud \qprobb = X_0$ holds for all $X \in \X$, $n \in \Natural$ and all stopping times $\tau$.
\end{itemize}
Using this reformulation, Theorem \ref{thm: main} bears more resemblance to the FTAP of \cite{MR1304434}. In fact, as already mentioned in the Introduction, in \cite{Kar_09_fin_add_ftap} Theorem \ref{thm: main} is used as an intermediate step in proving the FTAP in \cite{MR1304434}. 
\end{rem}

% \begin{rem} \label{rem: nec_suf in terms of pred_char}
% As the NA$_1$ condition is equivalent to NUPBR, one can use the results of \cite{MR2335830} to show that the conditions of Theorem \ref{thm: main} are  further equivalent to the finiteness of an explicit $[0, \infty]$-valued \emph{deterministic} functional of the predictable characteristics of $S$. This enables one to check the validity of NA$_1$ in a very straightforward way, which is not the case for NFLVR.
% \end{rem}

\begin{rem}
Theorem \ref{thm: main} is stated for one-dimensional semimartingales $S$, as even for this ``simple'' case the proof is quite technical and requires taking care of many different issues, as the reader will appreciate in Section \ref{sec: proof} below. There is no doubt that the result is still valid for the multi-dimensional semimartingale case, albeit its proof is expected to be significantly more involved.
\end{rem}

\section{The Proof of Theorem \ref{thm: main}} \label{sec: proof}

\subsection{Proving Theorem \ref{thm: main} with the help of an auxiliary result}
The proof of one implication of Theorem \ref{thm: main} is easy and somewhat classic, but will be presented anyhow here for completeness. Start by assuming the existence of an ELMD $Z$ and pick any sequence $(X_k)_{k \in \Natural}$ of wealth processes in $\X$ such that $\limk X_k(0) = 0$ as well as $X_k(T) \geq \xi$ for some nonnegative random variable $\xi$. Since $Z X_k$ is a nonnegative local martingale, thus a $\prob$-supermartingale,
\[
\expec[Z(T) \xi] \leq \expec[Z(T) X_k(T)] \leq Z(0) X_k(0) = X_k(0)
\]
holds for all $k \in \Natural$. Therefore, $\expec[Z(T) \xi] \leq 0$. Since $\prob[Z(T) > 0, \, \xi \geq 0] = 1$, $\expec[Z(T) \xi] \leq 0$ holds only if $\prob[\xi = 0] = 1$. Therefore, $\xi$ cannot be an arbitrage of the first kind, and condition NA$_1$ holds.

\smallskip
It remains to prove the other implication, which is considerably harder. Define
\[
\X_{++} \dfn \set{X \in \X \ | \ X > 0 \text{ and } X_- > 0}.
\]
Since condition NA$_1$ is equivalent to condition NUPBR of \cite{MR2335830}, the general results of the latter paper imply that condition NA$_1$ is equivalent to the existence of $\hX \in \X_{++}$ with $\hX(0) = 1$ such that, with $Z \dfn 1 / \hX$,  $Z X$ is a supermartingale for all $X \in \X_{++}$. (Note that the results of \cite{MR2335830} have been established when $S \in \X_{++}$; however, this condition is unnecessary. At any rate, in the present paper we give a full treatment instead of depending on results from \cite{MR2335830}.) Unfortunately, when jumps are present in $S$, these last supermartingales might fail to be local martingales. In order to achieve our goal, we shall have to slightly alter the original probability using the predictable characteristics of $S$. (The idea of how to perform such a change of probability is already present in \cite{MR1647282} and \cite{MR1804665}.) In \S\ref{subsec: dynamic case} below we shall establish the following result, certainly interesting in its own right. Before stating it, recall that for a signed measure $\mu$ on $(\Omega, \F)$, its \textsl{total variation} norm is defined as $\normTV{\mu} \dfn \sup_{A \in \F} \abs{\mu[A]}$.
%\smallskip
%It remains to prove the other implication, which is considerably harder. Clearly, it is enough to show the existence of a nonnegative process $Z$ with $Z(0) = 1$, $Z(T) > 0$, and such that $Z X$ is a local martingale for all $X \in \X_{++} \dfn \set{X \in \X \ | \ X > 0 \text{ and } X_- > 0}$. Now, since condition NA$_1$ is equivalent to condition NUPBR of \cite{MR2335830}, according to Theorem 4.12 of the latter paper, condition NA$_1$ is equivalent to the existence of $\hX \in \X_{++}$ with $\hX(0) = 1$ such that, with $Z \dfn 1 / \hX$,  $Z X$ is a supermartingale for all $X \in \X_{++}$. (The results of \cite{MR2335830} have been established when $S \in \X_{++}$. However, this condition is unnecessary; one can simply follow the development in \cite{MR2335830} working with $S$ directly, instead of the ``returns'' process $\int_0^\cdot (1 / S (t-)) \ud S(t)$ --- all the proofs carry through.) Unfortunately, these last supermartingales might fail to be local martingales. In order to achieve our goal, we shall have to slightly alter the original probability using the predictable characteristics of $S$. (The idea of how to perform such a measure change is already present in \cite{MR1647282} and \cite{MR1804665}.) In \S\ref{subsec: dynamic case} below we shall establish the following result, certainly interesting in its own right.
%
%

\begin{thm} \label{thm: help}
Assume that condition \NAone \ holds. Then, for any $\epsilon > 0$, there exists a probability $\tprob = \tprob(\epsilon)$ with the following properties:
\begin{enumerate}
  \item $\tprob$ is equivalent to $\prob$ on $\F(T)$.
  \item $\normTV{\tprob - \prob} \leq \epsilon$.
  \item There exists $\tX \in \X_{++}$ with $\tX(0) = 1$ such that $X / \tX$ is a \emph{local $\tprob$-martingale} for all $X \in \X$.
\end{enumerate}
\end{thm}

To see how Theorem \ref{thm: help} completes the proof of Theorem \ref{thm: main}, assume that condition NA$_1$ holds, as well as the statement of Theorem \ref{thm: help}. Define the process $Z$ via $Z_t \dfn (1 / \tX(t))( \ud \tprob / \ud \prob)|_{\F (t)}$ for $t \in \Real_+$, where $( \ud \tprob / \ud \prob)|_{\F (t)}$ denotes the Radon-Nikod\'ym derivative of $\tprob$ with respect to $\prob$ when the two probabilities are restricted on the sigma-algebra $\F(t)$. Then, Theorem \ref{thm: help}(1) implies that $Z(0) = 1$ and $Z(T) > 0$, and the fact that $Z X$ is a local martingale for all $X \in \X$ follows by Theorem \ref{thm: help}(3).

\subsection{The proof of Theorem \ref{thm: help}} \label{subsec: dynamic case}

In the course of the proof, results regarding the general theory of stochastic processes from \cite{MR1943877} are used. There are ideas from \cite{MR2335830} that are utilized throughout the proof; as the latter paper is long and technical, and in an effort to be as self-contained as possible, we are providing full arguments whenever possible. In fact, there is only one result from \cite{MR2335830} whose statement will just be assumed; this happens at the end of \S \ref{subsubsec: growth rates}.

\subsubsection{Predictable characteristics}

In order to prove Theorem \ref{thm: help},
we can assume without loss of generality that $S$ is a special semimartingale under $\prob$. Indeed, if this is not the case, we can change the original probability $\prob$ into another equivalent $\overline{\prob}$ using the Radon-Nikod\'ym density
\[
\frac{\ud \overline{\prob}}{\ud \prob} \dfn \frac{1}{\expec \bra{\pare{1 + \gamma \sup_{t \in \Real_+} |S(t)|}^{-1} }} \pare{1 + \gamma \sup_{t \in \Real_+} |S(t)|}^{-1},
\]
where $\gamma > 0$ is small enough so that $\normTV{\overline{\prob} - \prob} \leq \epsilon / 2$. Then, $\overline{\expec} \big[ \sup_{t \in \Real_+} |S(t)| \big] < \infty$, where ``$\overline{\expec}$'' denotes expectation under $\overline{\prob}$; in particular, $S$ is a special semimartingale under $\overline{\prob}$. Then, the validity of Theorem \ref{thm: help} can be shown for $\overline{\prob}$ and with $\epsilon / 2$ replacing $\epsilon$.

\smallskip

Now, assuming that $S$ is a special semimartingale under $\prob$, write its \emph{canonical} decomposition $S = S_0 + A + S^\co  + \int_{(0, \cdot] \times \Real} x \pare{\mu [\ud t, \ud x] - \nu [\ud t, \ud x]}$. Here, $A$ is \emph{predictable and of finite variation}, $S^\co$ is a local martingale with \emph{continuous} paths and $\int_{(0, \cdot] \times \Real} x \pare{\mu [\ud t, \ud x] - \nu [\ud t, \ud x]}$ is a
\emph{purely discontinuous} local
martingale. As usual, $\mu$ is the \textsl{jump measure} of $S$ defined via $\mu (D) := \sum_{t \in \Real_+} \indic_{D} (t, \Delta S(t)) \indic_{\Real \setminus \set{0} } (t)$, for $D \subseteq \Real_+ \times \Real$, and $\nu$ is the \textsl{predictable compensator} of the measure $\mu$. Since $S$ is a special semimartingale, we have $\int_{\Real_+ \times \Real} \pare{|x| \wedge |x|^2} \,\nu[\udw t, \udw x] < \infty$. We introduce the \textsl{quadratic covariation} process $C := [S^\co, S^\co]$ of $S^\co$, and define the predictable nondecreasing scalar process
\[
G \dfn C + \int_{(0, \cdot]} |\ud A(t)| + \int_{(0, \cdot] \times \Real} \pare{|x| \wedge |x|^2} \, \nu[\udw t, \udw x].
\]
All three processes $A$, $C$, and $\nu$ are absolutely continuous with respect to $G$. Therefore, we can write
\[
A= \int_{(0, \, \cdot]} a(t) \ud G(t), \ C = \int_{(0, \, \cdot]} c(t)  \ud G(t), \text{ and } \nu [(0, \cdot] \times E ] = \int_{(0, \, \cdot]} \kappa (t)[E] \ud G(t),
\]
where  $a$, $c$ and $\kappa$ are predictable, $a$ is a scalar process, $c$ a nonnegative scalar process, $\kappa$ a process with values in the set of measures on $(\Real, \B(\Real))$, where $\B(\Real)$ is the Borel sigma-algebra on $\Real$, that do not charge $\set{0}$ and integrate the function $\Real \ni x \mapsto |x| \wedge |x|^2$, and $E \in \B(\Real)$.

Condition NA$_1$ enforces some restrictions on the triplet of predictable characteristics of $S$. The next result is a consequence of \cite[Theorem 3.15(2)]{MR2335830}, but we provide the quick argument for completeness.
\begin{lem} \label{lem: consequences of na1}
Assume condition \NAone \ in the market. Then, with $\Lambda \dfn \Lambda_+ \cup \Lambda_-$, where
\begin{align*}
\Lambda_+ &\dfn \set{\kappa[(- \infty, 0)] = 0, \ c = 0, \ a > \int_{(0, \infty)} x \kappa [\ud x]} \text{ and} \\
\Lambda_- &\dfn \set{\kappa[(0, \infty)] = 0, \ c = 0, \ a < \int_{(-\infty, 0)} x \kappa [\ud x]},
\end{align*}
the predictable set $\Lambda$ is $(\prob \otimes G)$-null. (In particular, $\set{\kappa[\Real] = 0, \ c = 0, \ a \neq 0}$ is  $(\prob \otimes G)$-null.)
\end{lem}

\begin{proof}
Define $\vartheta \dfn \indic_{\Lambda_+} - \indic_{\Lambda_-}$. Then, it is straightforward to see that
\[
X^{0, \vartheta} = \int_{(0, \cdot]} \indic_{\Lambda} (t) \abs{a(t) - \int_{\Real} x \kappa(t) [\ud x]} \ud G(t) + \sum_{t \in (0, \cdot]} \indic_{\Lambda} (t) |\Delta S(t)|,
\]
where observe that the integral $\int_{\Real} x \kappa [\ud x]$ is always well defined on $\Lambda$. It is clear that $X^{0, \vartheta}$ is non-decreasing, i.e., $X^{0, \vartheta} \in \X$. Furthermore, if $\Lambda$ fails to be $(\prob \otimes G)$-null, then $\prob[X^{0, \vartheta} (T) > 0] > 0$. Let $\xi \dfn X^{0, \vartheta} (T)$, since $X^{x, \vartheta} (T) = x + \xi \geq \xi$ for all $x > 0$, $\xi$ is an arbitrage of the first kind. Therefore, under condition NA$_1$, $\Lambda$ has to be $(\prob \otimes G)$-null.
\end{proof}

\subsubsection{Changes of probability} \label{subsubsec: change of prob}

In what follows, a \textsl{strictly positive predictable random field} will refer to a function $Y : \Omega \times \Real_+ \times \Real \mapsto (0, \infty)$ that is measurable with respect to the product of the predictable sigma-algebra on $\Omega \times \Real_+$ with the Borel sigma-algebra on $\Real$. For any strictly positive predictable random field $Y$, let $\nu^Y$ be the predictable random measure that has density $Y$ with respect to $\nu$; in other words,
\begin{equation}  \label{eq: nuY}
\nu^Y [(0, \cdot] \times E ] = \int_{(0, \, \cdot]} \kappa^Y (t) [E] \ud G (t) = \int_{(0, \, \cdot]} \pare{ \int_E Y(t, x) \kappa (t)[\ud x] } \ud G(t)
\end{equation}
holds for all $E \in \B(\Real)$. For all $t \in \Real_+$, $Y(t, \cdot)$ is the density of $\kappa^Y (t)$ with respect to $\kappa (t)$.

Define the $(0, \infty)$-valued predictable process
\[
\eta \dfn \frac{\epsilon}{2 \abs{1 + G}^2},
\]
where we shall be assuming without loss of generality that $0 < \epsilon < 1$. In the sequel, we shall only consider strictly positive predictable random fields $Y$ such that the following properties are additionally identically satisfied:
\begin{enumerate}
  \item[(Y1)] $\int_\Real \pare{|x| \wedge |x|^2} \, \kappa^Y [\udw x]  <  \infty$.
  \item[(Y2)] $\int_\Real \abs{Y(x) - 1} \, \kappa[\udw x] \leq \eta$.
  \item[(Y3)] $\kappa[\Real] = \kappa^Y [\Real]$.
%  \item[(Y4)] $Y \geq 1/4$.
\end{enumerate}
(The dependence of processes on $(\omega, t) \in \Omega \times \Real_+$ is usually suppressed from notation to ease the reading. Whenever appropriate from the context, and for clarification purposes, we shall sometimes write $Y(x)$ or $Y(t,x)$ for $Y$.) 

Property (Y2) of $Y$ implies the estimate
\begin{eqnarray}
\nonumber  \int_{\Real_+ \times \Real} \abs{ Y(t, x) - 1 }  \nu[\udw t, \udw x] &=& \int_{\Real_+} \pare{\int_{\Real} \abs{ Y(t, x) - 1 } \, \kappa(t)[\udw x]} \ud G(t) \\
\label{eq: pre-hellinger}   &\leq& \int_{\Real_+} \eta(t) \ud G(t) \ = \ \frac{\epsilon}{2} \int_{\Real_+} \frac{\ud G(t)}{|1 + G(t)|^2} \ \leq \ \frac{\epsilon}{2}.
\end{eqnarray}
It follows that the process $M \dfn \int_{(0, \cdot] \times \Real} \pare{ Y (t, x) - 1 } \pare{ \mu [\udw t, \udw x] - \nu[\udw t, \udw x]}$ is a well defined local martingale. Observe that for all $t \in \Real_+$, we have
\[
\Delta M (t) = Y(t, \Delta S(t)) - 1 - \pare{\int_{\Real} \pare{Y(t,x) -1 } \kappa [\ud x]} \Delta G(t) = Y(t, \Delta S(t)) - 1 > - 1,
\]
holding in view of the fact that $Y$ is strictly positive and $\int_{\Real} \pare{Y(t,x) -1 } \kappa [\ud x] = \kappa^Y[\Real] - \kappa [\Real] = 0$, which follows from (Y3). With ``$\Exp$'' denoting the \textsl{stochastic exponential} operator, define 
\[
L \dfn \Exp (M) = \Exp \pare{ \int_{(0, \cdot] \times \Real} \pare{ Y (t, x) - 1 } \pare{ \mu [\udw t, \udw x] - \nu[\udw t, \udw x]}}.
\]
Combining \eqref{eq: pre-hellinger} with $\Delta M > -1$, a use of \cite[Theorem 12]{MR515738} gives that $L$ is a uniformly integrable martingale with $\prob [L(T) > 0] = 1$. However, because the last paper may be hard to obtain, we provide a quick argument in the present special case. At the same time, we show that the probability defined by $L$ satisfies requirement (2) of Theorem \ref{thm: help}.

\begin{lem} \label{lem: hellinger}
Let $Y$ be a strictly positive random field such that (Y1), (Y2) and (Y3) hold. With the above notation, we have $\prob [L(T) > 0] = 1$ and $\expec \bra{\sup_{t \in \Real_+} |L(t) - 1|} \leq \epsilon$. In particular, the recipe $\ud \prob^Y / \ud \prob = L(T)$ defines a probability $\prob^Y$ that is \emph{equivalent} to $\prob$ on $\F(T)$ such that $\normTV{\prob^Y - \prob} \leq \epsilon$.
\end{lem}

\begin{proof}
Since $\Delta M > -1$ and $M$ is a local martingale, $\prob [L(T) > 0] = 1$ follows.

Let $H \dfn \int_{(0^\cdot]} |Y(t, x) - 1| \mu [\ud t, \ud x]$ and $F \dfn \int_{(0^\cdot]} |Y(t, x) - 1| \nu [\ud t, \ud x]$. The process $F$ is the predictable compensator of $H$ and we have $\prob \bra{F (\infty) \leq \epsilon/2} = 1$ in view of \eqref{eq: pre-hellinger}. In particular, $M$ is a local martingale of finite variation.
%Now, observe that $M$ is actually a martingale with $\expec \bra{\sup_{t \in \Real_+} M_t} < \infty$. Indeed, this follows from the fact that
%\[
%\expec \bra{\sup_{t \in \Real_+} M_t} \leq \expec \bra{H_\infty + F_\infty} = 2 \expec \bra{F_\infty} \leq \frac{\epsilon^2}{4}.
%\] 
Using the fact that $L = 1 + \int_{(0, \cdot]} L (t -) \ud M (t)$, we obtain
\[
\expec \bra{\sup_{t \in \Real_+} \abs{L(t) - 1}} \leq \expec \bra{\int_{(0, \infty)} L (t-) \ud H (t) + \int_{(0, \infty)} L (t-) \ud F (t)} = 2 \expec \bra{\int_{(0, \infty)} L (t-) \ud F (t)}.
\] 
Furthermore, with $(\tau_n)_{n \in \Natural}$ being a localizing sequence for $L$, we have
\[
\expec \bra{\int_{(0, \tau_n]} L (t-) \ud F (t)} = \expec \bra{L (\tau_n) F (\tau_n)} - \expec \bra{\int_{(0, \tau_n]} F (t) \ud L (t)} \leq \frac{\epsilon}{2} \expec \bra{L (\tau_n)} \leq \frac{\epsilon}{2}.
\]
As the previous is valid for all $n \in \Natural$, $\expec \bra{\sup_{t \in \Real_+} \abs{L (t) - 1}} \leq \epsilon$ follows from a straightforward application of the monotone convergence theorem. In particular, $\expec \bra{\sup_{t \in \Real_+} \abs{L (t)}} < \infty$ which implies that $L$ is a uniformly integrable martingale and, therefore, $\prob^Y$ is well defined and equivalent to $\prob$ on $\F(T)$.  Furthermore, $\normTV{\prob^Y - \prob} = \expec \bra{\abs{L(T) - 1}} \leq \epsilon$, which completes the proof.
\end{proof}

%Continuing, \eqref{eq: pre-hellinger} coupled with the inequality $\abs{\sqrt{w} - 1}^2 \leq |w - 1|$, valid for all $w \in \Real_+$, gives
%\begin{equation} \label{eq: hellinger}
%  \int_{\Real_+ \times \Real} \abs{ \sqrt{Y(t, x)} - 1 }^2 \nu[\udw t, \udw x] \leq \frac{\epsilon^2}{8}.
%\end{equation}
%Inequality \eqref{eq: hellinger} in particular implies that
%\[
%\Delta \pare{\int_{(0, \cdot] \times \Real} \abs{ \sqrt{Y (t, x)} - 1 }^2  \nu[\udw t, \udw x]} = \pare{\int_{\Real} \abs{ \sqrt{Y (x)} - 1 }^2  \kappa [\udw x]} \Delta G \leq \frac{\epsilon^2 \Delta G}{8 |1 + G|^2} \leq 1,
%\]
%where recall that $0 < \epsilon < 1$. Therefore, combining the estimate \eqref{eq: hellinger}, \cite[IV.1.39, page 237]{MR1943877} and \cite[V.4.22, page 315]{MR1943877}, we obtain
%\[
%\normTV{\prob^Y - \prob} \ \leq \ 4 \sqrt{\expec \bra{\frac{1}{2} \int_{(0, T] \times \Real} \abs{ \sqrt{Y (t, x)} - 1 }^2  \nu[\udw t, \udw x] }} \ \leq \ 4 \sqrt{\frac{\epsilon^2}{16}} \ = \ \epsilon.
%\]
%It follows that any $Y$ satisfying (Y1), (Y2) and (Y3) generates a probability $\prob^Y$ satisfying requirements (1) and (2) of Theorem \ref{thm: help}. We shall soon see how to choose $Y$ so that requirement (3) of Theorem \ref{thm: help} is also satisfied.

Consider the probability $\proby$ of Lemma \ref{lem: hellinger}. According to Girsanov's Theorem (Theorem III.3.24, page 172 of \cite{MR1943877}),  under assumptions (Y1), (Y2) and (Y3) on $Y$, $S$ is still a special semimartingale under $\proby$ with canonical decomposition $S = S_0 + A^Y + S^{\co, Y}  + \int_{(0, \cdot] \times \Real} x (\mu[\ud t, \ud x] - \nu^Y[\ud t, \ud x])$, where the predictable compensator $\nu^Y$ of $\mu$ under $\proby$ was defined previously in \eqref{eq: nuY}, and where $A^Y = \int_{(0, \, \cdot]} a^Y(t) \ud G(t)$, with $a^Y \dfn a + \int_\Real x (Y(x) -1 ) \, \kappa [\udw x]$. For the continuous local $\prob^Y$-martingale part $S^{\co, Y}$ we have $C^Y := [S^{\co, Y}, S^{\co, Y}] = [S^{\co}, S^{\co}] = C$, i.e., $C^Y = \int_{(0, \, \cdot]} c^Y(t) \ud G(t)$ with $c^Y = c$.

\subsubsection{Relative rate of return} \label{subsubsec: rrr}

Remember that $Y$ always denotes a strictly positive predictable random field satisfying (Y1), (Y2), and (Y3) of \S\ref{subsubsec: change of prob}. We aim at understanding what extra condition must $Y$ satisfy in order for $\tprob \equiv \proby$ to satisfy all the requirements of Theorem \ref{thm: help}.

Define a pair of processes $(\ell, r)$ via
\[
\ell \dfn \inf \set{p \in \Real \such \kappa[\set{x \in \Real \such 1 + p x < 0}] = 0} \text{ and } r \dfn \sup \set{p \in \Real \such \kappa[\set{x \in \Real \such 1 + p x < 0}] = 0}.
\]
($\ell$ and $r$ are mnemonics for ``left'' and ``right'' respectively.) It is straightforward that $\ell \leq 0 \leq r$, as well as that both $\ell$ and $r$ are predictable: for example, $\set{\ell \leq p} = \Omega \times \Real_+$ if $p \in \Real_+$, while
\[
\set{\ell \leq p} = \bigcap_{n \in \Natural} \Big\{ \kappa[\set{x \in \Real \such 1 + (p + 1/n ) x < 0}] = 0 \Big\} \text{ if } p \in \Real \setminus \Real_+;
\]
in both cases, $\set{\ell \leq p}$ is predictable. Of course, nothing changes in the definition of $\ell$ and $r$ if we replace $\kappa$ with $\kappa^Y$. Define $I := [\ell, r] \cap \Real$.
%Observe that for any closed and bounded interval $J = [\inf J, \sup J] \subseteq \Real$, we have $\set{I \cap J = \emptyset} = \set{\sup J < \ell} \cup \set{r < \inf J} \in \Pre$; in other words, $I$ is a predictable process taking values in the closed subintervals of $\Real$ containing $\set{0}$.
Note that  $\cosupp(\kappa) = [-1/r, -1/ \ell] \cap \Real$, where ``$\cosupp$'' denotes the convex hull of the support of a measure.

For two $I$-valued predictable processes $p$ and $p'$, define a predictable process
\begin{equation} \label{eq: rel_perf}
\rely(p \such p') \dfn (p - p') \pare{a^Y - p' c^Y  -  \int_{\Real} \frac{p' |x|^2}{1 + p'x} \, \kappa^Y [\ud x] }.
\end{equation}
The last expression is closely related to the \emph{relative rate of return} of wealth processes in $\X_{++}$, as the proof of the following result reveals.
\begin{lem} \label{lem: rrr}
%Assume \NAone.
Suppose that $Y$ is a strictly positive predictable random field satisfying \emph{(Y1)}, \emph{(Y2)}, and \emph{(Y3)}. Further, suppose that $\tp$ is an $I$-valued predictable, $S$-integrable process such that $\rely(p \such \tp) = 0$ holds for all other $I$-valued predictable processes $p$. Define $\tX \dfn \Exp(\int_0^\cdot \tp(t) \ud S(t))$. Then, $\tX_0 = 1$, $\tX \in \X_{++}$, and $X / \tX$ is a local $\proby$-martingale for all $X \in \X$.
\end{lem}

\begin{proof}
Since $\tp$ is $S$-integrable, $\tX$ is well defined. In view of \eqref{eq: rel_perf}, the fact that $\rel(0\such \tp) = 0$ implies that $\kappa^Y[\set{x \in \Real \such \tp x = -1}] = 0$. Therefore, $\tp \Delta S > -1$, i.e., $\tX > 0$ and $\tX_- > 0$ hold. With $\tvartheta \dfn \tp \tX_-$, we have $\tX = X^{1, \tvartheta}$ in the notation of \S \ref{eq: wealth process, all}. Therefore, $\tX \in \X_{++}$.

Pick any $X = X^{x, \vartheta} \in \X_{++}$. Let $p \dfn \vartheta / X_-$; then, $X = x \Exp(\int_0^\cdot p(t) \ud S(t))$. We shall show that
\[
\frac{X}{\tX} = x \, \frac{\Exp(\int_0^\cdot p(t) \ud S(t))}{\Exp(\int_0^\cdot \tp(t) \ud S(t))}
\]
is a local $\proby$-martingale. Since $X > 0$, $X_- > 0$, $\tX > 0$, and $\tX_- > 0$ hold, it follows that we can write $X / \tX = x \Exp(R^{p \such \tp})$ for some semimartingale $R^{p \such \tp}$ with $\Delta R^{p \such \tp} > -1$. In fact,
\[
R^{p  \such \tp} = \int_0^\cdot  \pare{p(t) - \tp(t)} \ud S(t) - \int_0^\cdot \pare{p(t) - \tp(t)} \tp(t) \ud [S^\co, S^\co] (t) - \sum_{t \leq \cdot} \frac{\pare{p(t) - \tp(t)} \tp(t) |\Delta S(t)|^2}{1 +\tp(t) \Delta S(t)};
\]
indeed, using Yor's formula it can be easily checked that
\begin{align*}
\Exp \pare{ \int_0^\cdot \tp(t) \ud S(t) } \Exp \pare{ R^{p  \such \tp} } &= \Exp \pare{ \int_0^\cdot \tp(t) \ud S(t) + R^{p  \such \tp} + \bra{\int_0^\cdot \tp(t) \ud S(t), R^{p  \such \tp}}} \\
&= \ldots = \Exp \pare{ \int_0^\cdot p(t) \ud S(t) }.
\end{align*}
By a comparison of \eqref{eq: rel_perf} with the formula for $R^{p  \such \tp}$ above, $\rely(p \such \tp) = 0$ implies that $R^{p  \such \tp}$ is a sigma $\proby$-martingale. (For information and properties of sigma-martingales, the reader is referred to \cite{MR2013413}.) Since $X / \tX = x \Exp(R^{p  \such \tp})$, it follows that $X / \tX$ is a sigma $\proby$-martingale. For nonnegative processes, the sigma martingale property is equivalent to the local martingale property; therefore, we conclude that $X / \tX$ is a local $\proby$-martingale. 

Now, let $X \in \X$. Since $(1 + X) \in \X_{++}$, the discussion of the previous paragraph implies that $(1 + X) / \tX$ is a local $\proby$-martingale. Again, by the discussion of the previous paragraph, $1 / \tX$ is a local $\proby$-martingale. It follows that $X / \tX$ is a local $\proby$-martingale.
\end{proof}

%The next result, which we shall not prove, is a consequence of Theorems 3.15 and 4.12 in \cite{MR2335830}. One just has to use the fact that condition NA$_1$ is equivalent to condition NUPBR of \cite{MR2335830}.
%\begin{thm} \label{thm: exist of num}
%Assume \NAone. Let $Y$ be any strictly positive predictable random field such that (Y1) and (Y2) hold, and let $\proby$ be the induced probability. Then, there exists an $I$-valued, $S$-integrable predictable process $\hp^Y$ such that $\rely(p \such \hp^Y) \leq 0$ for any other $I$-valued predictable process $p$.
%\end{thm}

In view of Lemma \ref{lem: rrr}, Theorem \ref{thm: help} will be proved if we can find a strictly positive predictable random field $Y$ satisfying (Y1), (Y2) and (Y3), as well as an $I$-valued predictable, $S$-integrable process $\tp^Y$ such that $\rely(p \such \tp^Y) = 0$ holds for any other $I$-valued predictable process $p$. In \S \ref{subsubsec: growth rates}, we shall see how $\tp^Y$ should be picked, given a strictly positive predictable random field $Y$ satisfying (Y1), (Y2) and (Y3); then, in \S \ref{subsubsec: construction of density}, we shall construct the appropriate strictly positive predictable random field.

\subsubsection{Growth rates} \label{subsubsec: growth rates}

In order to understand how $Y$ has to be picked, we shall use the fact that the relative rate of return is essentially the directional derivative of the growth rate. In more detail, define a predictable random field $\gy$ via $\gy (\p) \dfn \p a^Y - (1/2) c^Y |\p|^2 - \int_{\Real} \pare{\p  x - \log(1 + \p x)}\, \kappa^Y [\udw x]$ for $\p \in I$, and set $\gy (\p) = - \infty$ when $\p \notin I$. The assumption $\int_{\Real} \pare{|x| \wedge |x|^2} \kappa^Y [\udw x] < \infty$ ensures that $\g$ is well-defined and finite in the interior of $I$, thought it might be the case that $\gy(\ell) = - \infty$ or $\gy(r) = - \infty$. It is obvious that for fixed $(\omega, t) \in \Omega \times \Real_+$, $\gy(\omega, t, \cdot) : \Real \mapsto \Real \cup \set{- \infty}$ is a concave function. With all set-inclusions involving subsets of $\Omega \times \Real_+$ from now on to be understood in a $(\prob \otimes G)$-a.e. sense, an application of Lemma \eqref{lem: consequences of na1} (with $a^Y$ and $\kappa^Y$ replacing $a$ and $\kappa$ there respectively) gives $\set{ r = \infty}= \set{\kappa^Y [(-\infty, 0)] = 0} \subseteq  \set{\lim_{\p \to \infty} \gy(\p) \leq 0}$. Indeed, $\set{\kappa^Y [(-\infty, 0)] = 0, \ c > 0} \subseteq \set{\lim_{\p \to \infty} \gy(\p) = - \infty}$, while $\set{\kappa^Y [(-\infty, 0)] = 0, \ c = 0} \subseteq \set{\lim_{\p \to \infty} \gy(\p) = a - \int_{(0, \infty)} x \kappa [\ud x]}$. Similarly, one can show that $\set{\ell = - \infty} \subseteq \set{\lim_{\p \to - \infty} \gy(\p) \geq 0}$. Since $\gy(0) = 0$, it follows that $\gy$ always achieves its supremum at some point in $I$.

Define now the ``derivative'' predictable random field $\ngoy : \Omega \times \Real_+ \times \Real \mapsto \Real \cup \set{- \infty, \infty}$ via
\begin{equation} \label{eq: growth der}
\ngoy (\p) \dfn a^Y - \p c^Y - \int_{\Real} \frac{\p |x|^2}{1 + \p x} \, \kappa^Y [\udw x] \, = \, \ngo(\p) + \int_\Real \frac{x}{1 + \p x}  \pare{Y(x) -1}  \kappa [\udw x],
\end{equation}
for $\p \in I$ (where $\ngo \equiv \ngo^1$), $\ngoy (\p) = \ngoy (\ell)$ for $\p < \ell$, and similarly $\ngoy (\p) = \ngoy (r)$ for $\p > r$. The concavity of $\gy$ and straightforward applications of the dominated convergence theorem imply that, for fixed $(\omega, t) \in \Omega \times \Real_+$, $\ngo$ is nonincreasing and continuous on $I$. Note that on $\set{\ell = 0 = r} = \set{\supp(\kappa) = \Real}$ it is impossible to define $\ngo$. In this case, we simply force $\ngoy(\p) = 0$ for all $\p \in \Real$; we shall see later how such convention is useful.

Define a process $\tp^Y \dfn \inf \set{\p \in I \such \ngoy (\p) \leq 0 }$, where we set $\tp^Y = r$ in case the last set is empty and $\tp^Y = 0$ on $\set{\ngoy (\ell) = 0 = \ngoy(r)}$. It is clear that $\tp^Y$ is a predictable process. Furthermore, on $\{ \ngoy (\ell) \geq 0, \, \ngoy (r) \leq 0 \}$, which is a predictable set, we have $\ngoy (\tp^Y) = 0$, and, therefore, $\rely(p \such \tp^Y) = (p - \tp^Y) \ngoy (\tp^Y) = 0$ for all $I$-valued predictable processes $p$.

The point of the above discussion is the following: Suppose that for some strictly positive predictable random field $Y$ satisfying (Y1), (Y2) and (Y3), both $\ngoy (\ell) \geq 0$ and $\ngoy (r) \leq 0$ hold for all $(\omega, t) \in \Omega \times \Real_+$, which as usual will be suppressed from notation in the sequel. Then, we can construct a predictable $I$-valued process $\tp^Y$ such that $\rely(p \such \tp^Y) = (p - \tp^Y) \ngoy (\tp^Y) = 0$ for all $I$-valued predictable processes $p$. (Observe how $\rely(p \such \tp^Y) = (p - \tp^Y) \ngoy (\tp^Y) = 0$ trivially also holds on $\set{\ell = 0 = r} = \set{\supp(\kappa) = \Real}$ in view of our convention, as $I = \set{0}$.) In view of Lemma \ref{lem: rrr}, Theorem \ref{thm: help} will follow as soon as we know that $\tp^Y$ is $S$-integrable. Luckily, this is \emph{always} the case under condition NA$_1$. The proof of this fact is quite technical, and basically follows the treatment in \cite[Section 8]{MR2335830}, where Proposition 4.16 of the latter paper is proved. We shall, however, provide some details for completeness. In view of \cite[Corollary 3.6.10, page 128]{MR1906715}, failure of $S$-integrability of $\tp^Y$ implies that there exist a sequence of $[0,1]$-valued predictable processes $(h_n)_{n \in \Natural}$, such that each $h_n \tp^Y$, $n \in \Natural$, is $S$-integrable and the sequence of terminal values $\pare{\int_0^T \pare{h_n(t) \tp^Y (t)} \ud S(t)}_{n \in \Natural}$ fails to be bounded in probability. (Note that, a priori, the previous sequence can fail to be bounded in probability either from above or below, or even from both sides.) For each $n \in \Natural$, define $X_n \in \X_{++}$ with $X_n(0) = 1$ via
\[
X_n \dfn \Exp \pare{\int_0^\cdot \pare{h_n(t) \tp^Y (t)} \ud S(t)}.
\]
Since $h_n$ is $[0,1]$-valued, the definition of $\tp^Y$ implies that $\rely(0 \,|\, h_n \tp^Y) \leq 0$. (This follows because the predictable function $[0,1] \ni u \mapsto \g(u \tp^Y)$ is nondecreasing.) Therefore, $1 / X_n$ is a nonnegative $\prob$-supermartingale for all $n \in \Natural$. Then, it follows from \cite[Lemma 8.1]{MR2335830} that failure of boundedness in probability of $\pare{\int_0^T \pare{h_n(t) \tp^Y (t)} \ud S(t)}_{n \in \Natural}$ also implies failure of boundedness in probability of the sequence $(X_n(T))_{n \in \Natural}$. (Although intuitively plausible, passing from failure of boundedness in probability of processes to failure of boundedness in probability of their stochastic exponentials is not always possible, because the stochastic exponential is not a monotone operator. The fact that this can be done in the present case is due to the fact that each process $1 / X_n$ is a nonnegative $\prob$-supermartingale --- see also \cite[Remark 8.2]{MR2335830}.) However, condition NA$_1$ is equivalent to the requirement that the set $\set{X(T) \such X \in \X \text{ with } X(0) = 1}$ is bounded in probability, making it impossible for $(X_n(T))_{n \in \Natural}$ to fail to be bounded in probability. We conclude that $\tp^Y$ is $S$-integrable under the validity of condition NA$_1$.

\subsubsection{Construction of the appropriate predictable random field} \label{subsubsec: construction of density}

We now move to the most technical part of the proof of Theorem \ref{thm: help}, by constructing a strictly positive predictable random field $Y$ satisfying (Y1), (Y2), and (Y3), as well as the following condition:
\begin{enumerate}
  \item[(Y4)] $\ngo^Y (\ell) \geq 0$ and $\ngo^Y (r) \leq 0$.
\end{enumerate}
(Note that the last condition is always trivially satisfied on $\set{\ell = 0 = r} = \set{\supp(\kappa) = \Real}$.) From the discussion of \S \ref{subsubsec: rrr} and \S \ref{subsubsec: growth rates}, existence of such a strictly positive predictable random field $Y$ will complete the proof of Theorem \ref{thm: help}.

The strictly positive predictable random field $Y$ will actually depend on the predictable processes $(a, \kappa, \eta)$ and will have to be defined differently on each of nine predictable sets $(P_i)_{i=1, \ldots, 9}$ that constitute a partition of $\Omega \times \Real_+$. (By construction, it will be immediately clear that $Y$ is actually a predictable random field.) On each of these predictable sets we shall show that (Y1) to (Y4) are valid. The reader will notice how the one-dimensional structure of the asset-price process is used in a non-trivial way when defining $Y$. The method certainly does not generalize for the case of multiple assets --- it appears a big challenge to provide a proof in a multi-dimensional setting.

Before we delve into the technicalities of the proof, recall that under condition NA$_1$, any strictly positive predictable random field $Y$ satisfying (Y1), (Y2) and (Y3) is such that $\set{\ell = - \infty} \subseteq \set{\ngoy (\ell) \geq 0}$ and $\set{r = \infty} \subseteq \set{\ngoy (r) \leq 0}$. This is true in view of Lemma \ref{lem: consequences of na1} --- see also the discussion in \S \ref{subsubsec: growth rates}.

\smallskip

\noindent $\bullet$ We start with the set $P_1 \dfn \set{\ell = 0, \, r = \infty}$. (All the predictable-set inclusions below are understood to hold on $P_1$, until we move to the next case where they will be understood to hold on $P_2$, and so forth.) Here, $\ngo(\ell) = \ngo(0) = a$. Since, as explained above, $\set{r = \infty} \subseteq \set{\ngo^Y (r) \leq 0}$, we only have to carefully define $Y$ on $\set{a < 0}$. Notice that $\set{\ell = 0, r = \infty} = \set{\cosupp(\kappa) = [0,\infty)}$, and define $Y_1 \dfn y_1(a, \kappa, \eta)$, where, with
\[
\delta \dfn 1 + \frac{4}{\kappa[\Real]} + \inf \set{x \in \Real \
 \Big| \ \kappa[(0, x]] \geq \frac{\kappa[\Real]}{2} } \text{ and } b \dfn \abs{\delta - a + \frac{2}{\eta}}^2,
\]
we set
\[
y_1(a, \kappa, \eta; \, x) \dfn 1 + \pare{\frac{1}{\sqrt{b} \, \kappa \bra{ (b,  \infty)}} \indic_{( b, \, \infty)} (x)
- \frac{1}{\sqrt{b} \, \kappa \bra{ (0, \delta]}} \indic_{(0, \delta]} (x) } \indic_{\set{a < 0}} \text{ for } x \in \Real,
\]
(In the definition of $y_1(a, \kappa, \eta)$, the term $1 / ( \sqrt{b} \, \kappa \bra{ (0, \delta]} )$  is understood to be zero on $\set{\kappa[\Real] = \infty}$.) We shall show below that $Y_1$ satisfies (Y1) through (Y5). On $\set{a \geq 0}$ this is trivial, since $Y_1 = 1$. Therefore, focus will be given only on $\set{a < 0}$ below. First of all, it is easy to see that $Y_1 \geq 1/2$. Indeed, on $\set{\kappa[\Real] = \infty }$ we have $Y_1 \geq 1$; also, on $\set{\kappa[\Real] < \infty}$,
\[
\sqrt{b} \, \kappa \bra{ (0, \delta]} > \delta \kappa \bra{ (0, \delta]} > \frac{4}{\kappa[\Real]} \, \frac{\kappa[\Real]}{2} = 2
\]
holds from the definition of $\delta$. Proceeding, the fact that $Y_1$ is bounded from above coupled with $\int_\Real \pare{|x| \wedge |x|^2} \, \kappa[\udw x] < \infty$ implies $\int_\Real \pare{|x| \wedge |x|^2} Y_1 (x) \, \kappa[\udw x] < \infty$. For the estimate of the distance between $\kappa$ and $\kappa^{Y_1}$ observe that
\[
\int_\Real |Y_1(x) - 1| \, \kappa[\udw x] \leq  \frac{2}{\sqrt{b}} \leq  \frac{2}{ 2 / \eta} = \eta.
\]
Now, on $\set{\kappa[\Real] = \infty}$ we have $Y_1 \geq 1$ and obviously $\kappa^{Y_1} [\Real] = \infty$; on the other hand, on $\set{\kappa[\Real] < \infty}$ the equality $\kappa^{Y_1} [\Real] = \kappa[\Real]$ follows in a straightforward way from the definition of $Y_1$. Finally, since $\ngo (0) = a$, use \eqref{eq: growth der} to estimate
\begin{eqnarray*}
  \ngo^{Y_1} (0) &=& a + \int_{(b,  \infty)} \frac{x}{\sqrt{b} \, \kappa \bra{ (b,  \infty)}} \kappa[\udw x] - \int_{(0, \delta]} \frac{x}{\sqrt{b} \, \kappa \bra{ (0, \delta]}} \kappa[\udw x]   \\
   &\geq& a + \sqrt{b} - \frac{\delta}{\sqrt{b}} \ = \ a - a + 2 / \eta + \delta - \frac{\delta}{\delta - a + 2 / \eta} \ \geq \ 0.
\end{eqnarray*}
(The last inequality follows from $\eta > 0$ and $\delta > 1$, which imply also $\delta - a + 2 / \eta > 1$, since $a < 0$.)

\smallskip

\noindent $\bullet$ The situation on $P_2 \dfn \set{\ell = - \infty, \, r = 0}$ is symmetric to the previous one. With
\[
\delta \dfn 1 + \frac{4}{\kappa[\Real]} - \sup \set{ x \in \Real \ \Big| \ \kappa[[x, 0)] \geq \frac{\kappa[\Real]}{2} } \text{ and } b \dfn \abs{\delta + a + \frac{2}{\eta}}^2,
\]
define $Y_2 \dfn y_2(a, \kappa, \eta)$, where
\[
y_2(a, \kappa, \eta; \, x) \dfn 1 + \pare{\frac{1}{\sqrt{b} \, \kappa \bra{ ( - \infty, \, - b)}} \indic_{( - \infty, \, - b)} (x)
- \frac{1}{\sqrt{b} \, \kappa \bra{ [- \delta, 0)}} \indic_{[- \delta, 0)} (x)} \indic_{\set{a > 0}} \text{ for } x \in \Real.
\]
One can then follow the exact same steps that we carried out on $P_1$.

\smallskip

\noindent $\bullet$ We now move to the set $P_3 \dfn \set{\ell = - \infty, \, 0 < r < \infty}$, on which $\cosupp(\kappa) = [-1/r, 0]$. Since $\ell = - \infty$, we have $\ngo(\ell) \geq 0$. Also, on $\set{\kappa[\set{-1 / r}] > 0}$ we have $\g(r) = - \infty$, and $\ngo(r) = - \infty$ follows easily. Then, define $Y_3 \dfn y_3(a, \kappa, \eta)$, where, with 
\[
\beta \dfn \frac{1}{r} \min \set{ \frac{1}{2}, \, \exp \pare{- \frac{2 r}{\kappa[\Real]}}, \, \exp \pare{- \frac{2 r}{\eta}}},
\]
$y_3(a, \kappa, \eta; \, x)$ is for all $x \in \Real$ equal to
\[
1 + \pare{\frac{r}{\kappa[\Real] \log (r \beta) }
+ \indic_{(-\frac{1}{r}, \, \beta -\frac{1}{r}]} (x) \int_{x}^{\beta -\frac{1}{r}} \frac{|r|^2}{(1 + rw) \, |\log (1 + rw)|^2 \, \kappa \bra{(-\frac{1}{r}, \, w]}} \ud w} \indic_{\set{\kappa \bra{\set{-\frac{1}{r}}] = 0}}}
\]
Since $\log(r \beta) \leq - 2 r / \kappa[\Real]$, we easily get $Y_3 \geq 1 / 2 > 0$. On $\set{\kappa[\Real] = \infty}$, $Y_3 \geq 1$ and $\kappa^{Y_3}[\Real] = \infty$ trivially follows; on the other hand, on $\set{\kappa[\Real] < \infty}$, $\kappa^{Y_3} [\Real] = \kappa[\Real]$ follows as long as one notices that the double integral
\[
\int_{(-1 / r, \, \beta - 1/r]} \pare{\int_{x}^{\beta - 1/r} \frac{|r|^2}{(1 + rw) \, |\log (1 + rw)|^2 \, \kappa \bra{(-1 / r, \, w]}} \ud w} \kappa [\udw x]
\]
is, in view of Fubini's theorem, equal to
\begin{equation} \label{eq: helpful estimate}
\int_{-1 / r}^{\beta - 1/r}  \frac{|r|^2}{(1 + rw) \, |\log (1 + r w)|^2} \ud w = r \int_{0}^{r \beta}  \frac{1}{w \, |\log w|^2} \ud w = - \frac{r}{\log (r \beta)}.
\end{equation}
The above estimate also implies $\int_\Real \pare{|x| \wedge |x|^2} Y_3(x) \, \kappa[\udw x] < \infty$. Indeed, note that
\[
Y_3(x) \leq 1 + r / (\kappa[\Real] \log (r \beta))
\]
for $x \in I \setminus (-1 / r, \, \beta - 1/r]$, while, using the fact that $\beta \leq 1 / (2 r)$, we obtain
\[
\int_{(-1 / r, \, \beta - 1/r]} \pare{|x| \wedge |x|^2} Y_3(x) \, \kappa[\udw x] \leq \frac{1}{r \min \set{1, r}} \int_{(-1 / r, \, \beta - 1/r]}  Y_3(x) \, \kappa[\udw x] < \infty.
\]
For estimating the distance between $\kappa$ and $\kappa^{Y_3}$, note that
\[
\int_\Real |Y_3(x) - 1| \, \kappa[\udw x] \leq - 2 r / \log (r \beta) \leq \eta,
\]
which follows from the definition of $\beta$ and the calculations that lead to \eqref{eq: helpful estimate}. We shall now show that $\g^{Y_3} (r) = - \infty$, therefore establishing that $\ngo^{Y_3}(r) \leq 0$. Start with the observation that, for $x \in (-1 / r, \, \beta - 1/r]$, integration by parts gives
\begin{eqnarray*}
  \log(1 + r x) Y_3 (x) &=& \log (r \beta) + \frac{r}{\kappa[\Real]} - \int_{x}^{\beta - 1/r} \frac{r}{1 + r w} Y_3 (w) \ud w + \\
    & & \int_{x}^{\beta - 1/r} \frac{|r|^2}{(1 + r w) \, \log (1 + r w) \, \kappa \bra{(-1 / r, \, w]}} \ud w \\
    &\leq& \frac{r}{\kappa[\Real]} + \int_{x}^{\beta - 1/r} \frac{|r|^2}{(1 + r w) \, \log (1 + r w) \, \kappa \bra{(-1 / r, \, w]}} \ud w. \\
\end{eqnarray*}
The above estimate and Fubini's theorem imply that $\int_{(-1/r, \, \beta - 1/r]} \log(1 + r x) Y_3 (x) \, \kappa[\udw x]$ is bounded from above by the quantity 
\[
\frac{r \kappa[(-1/r, \, \beta - 1/r]]}{\kappa[\Real]} + |r|^2 \int_{-1/r}^{\beta - 1/r} (1 + rw)^{-1} \, \log^{-1} (1 + rw) \ud w = - \infty.
\]
This last fact, together with \eqref{eq: growth der} and $\int_\Real \pare{|x| \wedge |x|^2} \, \kappa[\udw x] < \infty$ gives $\g^{Y_3} (r) = - \infty$. Of course, $\ngo^{Y_3} (\ell) \geq 0$ follows because $\ell = - \infty$.

\smallskip

\noindent $\bullet$ The situation on $P_4 \dfn \set{- \infty < \ell < 0, \, r = \infty}$ is symmetric to $P_3$ and, therefore, details will be omitted. Just define $Y_4 \dfn y_4 (a, \kappa, \eta)$, where, with 
\[
\beta \dfn  \frac{1}{\ell} \min \set{ \frac{1}{2} , \, \exp \pare{ \frac{2 \ell}{\kappa[\Real]}}, \, \exp \pare{ \frac{2 \ell}{\eta}}},
\]
$y_4 (a, \kappa, \eta; \, x)$ is for all $x \in \Real$ equal to
\[
1 + \pare{\frac{\ell}{\kappa[\Real] \log (\ell \beta) }
+ \indic_{(\beta - \frac{1}{\ell}, \, - \frac{1}{\ell}]} (x) \int_{\beta - \frac{1}{\ell}}^{x} \frac{|\ell|^2}{(1 + \ell w) \, |\log (1 + \ell w)|^2 \, \kappa \bra{[w, - \frac{1}{\ell})}} \ud w} \indic_{\set{\kappa \bra{\set{-\frac{1}{\ell}}] = 0}}}.
\]
\smallskip

\noindent $\bullet$ We now move to $P_5 \dfn \set{\ell = 0, \, 0 < r < \infty}$. Here, we shall use a combination of the work we carried out for $P_1$ and $P_3$. Remembering the definitions of the deterministic functionals $y_1$ and $y_3$, define
\[
Y_5 \dfn y_1 \pare{a^{y_3 (a, \kappa, \eta/2)}, \kappa^{y_3 (a, \kappa, \eta/2)}, \, \eta / 2} \, y_3 (a, \kappa, \eta/2).
\]
The definition of $Y_5$ is essentially realized in two steps. First there is a change according to $y_3$. This forces $\g^{y_3 (a, \kappa, \eta/2)} (r) = - \infty$ as on $P_3$. Also, (Y1), (Y2) and (Y3) hold, with $\eta / 2$ replacing $\eta$ in (Y2). In the second step there is a change using $y_1$. Since $y_1(a^{y_3 (a, \kappa, \eta/2)}, \kappa^{y_3 (a, \kappa, \eta/2)}, \, \eta / 2; x) = 1$ for all $x \in (- \infty, 0)$, $\g^{Y_5} (r) = - \infty$ (and, therefore, $\ngo^{Y_5}(r) \leq 0$) still holds, while now it is also the case that $\ngo^{Y_5} (\ell) \geq 0$, as was the case on $P_1$. It is clear that $Y_5 > 0$ (since both of the predictable random fields appearing in the definition of $Y_5$ are strictly positive), and that (Y1) to (Y4) all hold.

\smallskip

\noindent $\bullet$ On $P_6 \dfn \set{- \infty < \ell < 0, \, r = 0}$, define
\[
Y_6 \dfn  y_2 \pare{a^{y_4 (a, \kappa, \eta/2)}, \kappa^{y_4 (a, \kappa, \eta/2)}, \, \eta / 2} \, y_4 (a, \kappa, \eta/2).
\]
The situation is symmetric to the one on $P_5$ --- just follow the exact same reasoning.

\smallskip

\noindent $\bullet$ Moving to $P_7 \dfn \set{- \infty < \ell < 0 < r < \infty}$, we shall use a combination of the treatment on $P_3$ and $P_4$. Define
\[
Y_7 \dfn y_3 \pare{a^{y_4 (a, \kappa, \eta/2)}, \kappa^{y_4 (a, \kappa, \eta/2)}, \, \eta / 2} \, y_3 (a, \kappa, \eta/2).
\]
The validity of (Y1), (Y2), (Y3) and (Y4) follow by the same reasoning carried out on the set $P_5$.

\smallskip

\noindent $\bullet$ On $P_8 \dfn \set{\ell = 0, \, r = 0} \subseteq \set{\ngo(0) = 0}$ there is no need to do anything: simply set $Y_8 \dfn 1$.

\smallskip

\noindent $\bullet$ Finally, on $P_9 \dfn \set{\ell = - \infty, \, r = \infty} = \set{\cosupp(\kappa) = \emptyset}$ there is also no need to do anything; set $Y_9 \dfn 1$. Indeed, we either have $c = 0$, which implies that $a =0$ and, therefore, $\ngo(- \infty) = \ngo(+\infty) = 0$, or $c > 0$, in which case $\ngo(- \infty) = \infty$ and $\ngo(+\infty) = -\infty$.

%----------------------------------------------------------------
\bibliographystyle{siam}
\bibliography{na1}
\end{document}